\documentclass[twocolumn,prl]{revtex4}
\usepackage{amsmath}
\usepackage{amsthm}
\usepackage{graphicx}

\newtheorem{lemma}{Lemma}

\newcommand{\ket}[1]{\left|{#1}\right\rangle}
\newcommand{\bra}[1]{\left\langle {#1}\right|}

\begin{document}

\title{Quantum communication beyond the localization length in disordered spin chains}

\author{Jonathan Allcock}
\email[Electronic address: ]{jon.allcock@bristol.ac.uk}

\author{Noah Linden}
\email[Electronic address: ]{n.linden@bristol.ac.uk}
\affiliation{Department of Mathematics, University of Bristol, University Walk, Bristol BS8 1TW, United Kingdom}

\date{January 27, 2008}

\begin{abstract}
We study the effects of localization on quantum state transfer in spin chains. We show how to
use quantum error correction and multiple parallel spin chains to send a qubit with high fidelity over arbitrary 
distances; in particular distances much greater than the localization length of the chain. 
\end{abstract}

%\pacs{}
%\keywords{}

\maketitle

The reliable communication of quantum states from one physical
location to another is most likely necessary if large-scale
quantum computing is ever to be realised. Several years ago Bose
\cite{Bos03} proposed using spin chains as the medium for such
quantum state transfer.  Since then, there has been much
interest in this area, and numerous protocols have been put
forward which develop this idea, improving the efficiency and
fidelity of transfer using methods such as encoding the
information over multiple qubits in the chain
\cite{OL04,Has05,BGB06, Bur07}, using engineered couplings
\cite{CDEL04,CDDEKL05}, multi-rail encodings
\cite{BB05a,SJBB07,BGB05} and local memory \cite{GB06}.

Unfortunately, any real spin chain will inevitably have an element
of disorder inherent in the system.  As pointed out in
\cite{KLMW06}, this can cause a phenomenon now known as Anderson localization \cite{And58} to take place.  This is the
process in which the energy eigenstates of a disordered lattice
become localized in space, rather than extending throughout
the system as they would in the absence of disorder. This is turn inhibits state transfer beyond a distance known as the localization length of the chain. While this may be useful for localizing qubits for quantum computing \cite{DSSI05}, it provides a major challenge for the use of spin chains in quantum communication. Although localization has been extensively studied in solid state physics, its implications for quantum information are not well understood \cite{KLMW06},\cite{Bur07},\cite{BB05b,CRMF05, BO07}. 

Indeed, it is not clear at first glance how the problems due to
localization fit within the standard error paradigm
considered in quantum information theory. Firstly, rather than
being due to some coupling of the spin chain with the
environment, localization is an intrinsic source of error. Even in the absence of disorder, the excitations carrying the quantum
information become spread out along the chain.  In the presence of
disorder, only an exponentially small part of the signal reaches
beyond the localization length.   Also, although we are only trying to communicate a single
qubit, the localization errors take place in the
higher-dimensional Hilbert space of the entire spin chain.

In this Letter we look more closely at the effect localization has
on spin chain state transfer and find that, for a class of
standard spin chain protocols, localization can effectively be
viewed as a source of amplitude damping errors, where the damping
parameter is dependent on the distance propagated and the degree
of disorder in the chain.  We then show how to use multiple
spin chains and quantum error correction to achieve high fidelity
quantum information transfer over arbitrary distances; in
particular over distances much greater than the localization
length. By considering a concatenation scheme we show that,
provided the disorder is not too great, the number of spin chains
required scales only polylogarithmically with the distance over
which we wish to communicate.

In what follows we shall use the following notation and
conventions: A spin in the $\ket{1}$ state will be called
 an \textit{excitation}. For a system of $N$ spins, we shall
use a bold font $\ket{\textbf{0}}= \ket{00...0}$ to denote the
zero excitation state, and $\ket{\textbf{j}}= \ket{00...010...0}$
to denote the single excitation state with the $j^{\text{th}}$
spin in the state $\ket{1}$ and all the other spins in state
$\ket{0}$. Finally, $\vec{\sigma_{i}} =
(\sigma_{i}^{x},\sigma_{i}^{y},\sigma_{i}^{z})$, where
$\sigma_{i}^{x,y,z}$ are the Pauli matrices acting on the
$i^{\text{th}}$ spin in the chain.

Let us consider the following communication scenario. Alice and
Bob are at opposite ends of a chain of $N$ spin-$1/2$ particles
described by some nearest-neighbour Hamiltonian $H$. Alice's task
is to send, with as high a fidelity as possible, an unknown qubit
state to Bob.  We shall assume that the following conditions hold:
$(\text{C1})$: The system is isolated from the environment, and
thus there are no external sources of noise. $(\text{C2})$: $H$
commutes with the total $Z$-spin operator
$\sum_{k}^{N}\sigma_{k}^{z}$ and hence conserves the number of
excitations on the chain. $(\text{C3})$: The system starts in the
initial state $\ket{\textbf{0}}$. For example, $H$ could be a
simple Heisenberg coupling (in the absence of an external field)
\begin{equation}\label{e:hamil}
H   =  -(J/2)\sum_{<i,j>}\vec{\sigma_{i}}\cdot\vec{\sigma_{j}}
\end{equation}
where $\sum_{<i,j>}$ denotes that the sum is over all adjacent spins $i$ and $j$, and $J$ is the coupling constant between spins. Note that we have not yet introduced any disorder into the system.

The communication proceeds according to \cite{OL04} and
\cite{Has05}. We assume that Alice and Bob each have access to a
number $N_A,\ N_B$ respectively, of spins at their ends of
the chain. To begin, Alice encodes the input state $a\ket{0}
+ b\ket{1}$ as a state of her $N_A$ spins in such a way that
the following two conditions hold.  $(\text{C4})$: $\ket{0}$ is
encoded as $\ket{\textbf{0}}_{A}$. $(\text{C5})$: $\ket{1}$ is
encoded as $\ket{1_{\text{ENC}}}_{A}$. Here $A$ denotes Alice's
addressable spins, $\overline{A}$ denotes all other spins, and
$\ket{1_{\text{ENC}}}_{A}$ lies in the linear span of
$\ket{\textbf{j}}_{A}$, i.e. it is a superposition of single
excitation states, where the excitations lie in Alice's domain.
The system then undergoes unitary time evolution for some time
$t$. Note that as a result of Alice's encoding, and the fact that
the Hamiltonian preserves the total number of excitations, the
chain dynamics remain restricted to the $N+1$ dimensional subspace
spanned by the zero and single excitation states. Finally, Bob
applies a decoding unitary to his addressable spins. This
concentrates the state onto a single spin which he then takes as
the output of the transfer. The whole process is equivalent to
sending Alice's original state down an amplitude damping channel
with time-dependent channel parameter $\gamma(t) =
1-\mathcal{C}_{B}(t)$, where $0 \leq C_{B}\left(t\right)\leq 1$.
The corresponding average fidelity is given by $1/2 +
\sqrt{C_{B}\left(t\right)}/3 + C_{B}\left(t\right)/6$.

There is a convenient way to visualise the transfer
process \cite{OL04}.   At time $t$ the state of the system can be
expressed as $a\ket{\textbf{0}} + b\sum_{j =
1}^{N}c_{j}(t)\ket{\textbf{0}}$, where $\sum_{j =
1}^{N}\left|c_{j}(t)\right|^{2} = 1$.  By plotting the quantities
$\left|c_{j}(t)\right|^{2}$ against site number $j$ we can produce
a graph of the state.  Then, the quantity $\mathcal{C}_{B}(t)$ is
given by the area under the graph supported by Bob's accessible
sites. To achieve high fidelity transfer, the strategy is to
choose $\ket{1_{\text{ENC}}}_{A}$ to be a wavepacket with a
particular shape that leads to minimal dispersion.

What happens now if there is disorder present in the chain? Any real chain will, due to engineering limitations and thermal fluctuations, have spin-spin couplings that are not described by an ideal Hamiltonian such as \eqref{e:hamil}. In general there may be complicated time-dependent perturbations to the Hamiltonian, possibly coupling the zero and single excitation subspaces with subspaces of larger numbers of excitations.  However, as in \cite{KLMW06}, here we consider a simpler model of randomness where the total Hamiltonian $H_{\epsilon}$ is given by
\begin{equation}\label{e:perturbed}
H_{\epsilon} = H + \sum_{j}\epsilon_{j}\ket{\textbf{j}}\bra{\textbf{j}}
\end{equation}
where $H$ is the Hamiltonian \eqref{e:hamil} and the $\epsilon_{j}$ are i.i.d real random variables drawn from some underlying distribution $\mathcal{P}_{\delta}$ with bounded density, characterised by a disorder parameter $\delta$ (for instance, the uniform distribution in the interval $\left[-\delta,\delta\right]$).

What are the implications of this for state transfer? In spite of
the disorder, the all zero state $\ket{\textbf{0}}$ trivially
remains an eigenstate of $H_{\epsilon}$ and, furthermore, the
chain evolution remains restricted to the zero and single
excitation subspaces.  This implies that the disordered chain
still behaves like an amplitude damping channel, where the damping
parameter depends on the area under the graph in Bob's domain. Fig
\ref{fig:locprop} shows how increasing the degree of disorder causes
the graph of the state to suffer from dispersion and reflection as
it propagates.  Consequently, the area under the graph in Bob's
domain, and hence the transfer fidelity, becomes increasingly
suppressed.
\begin{figure}
        \includegraphics[scale = 0.5]{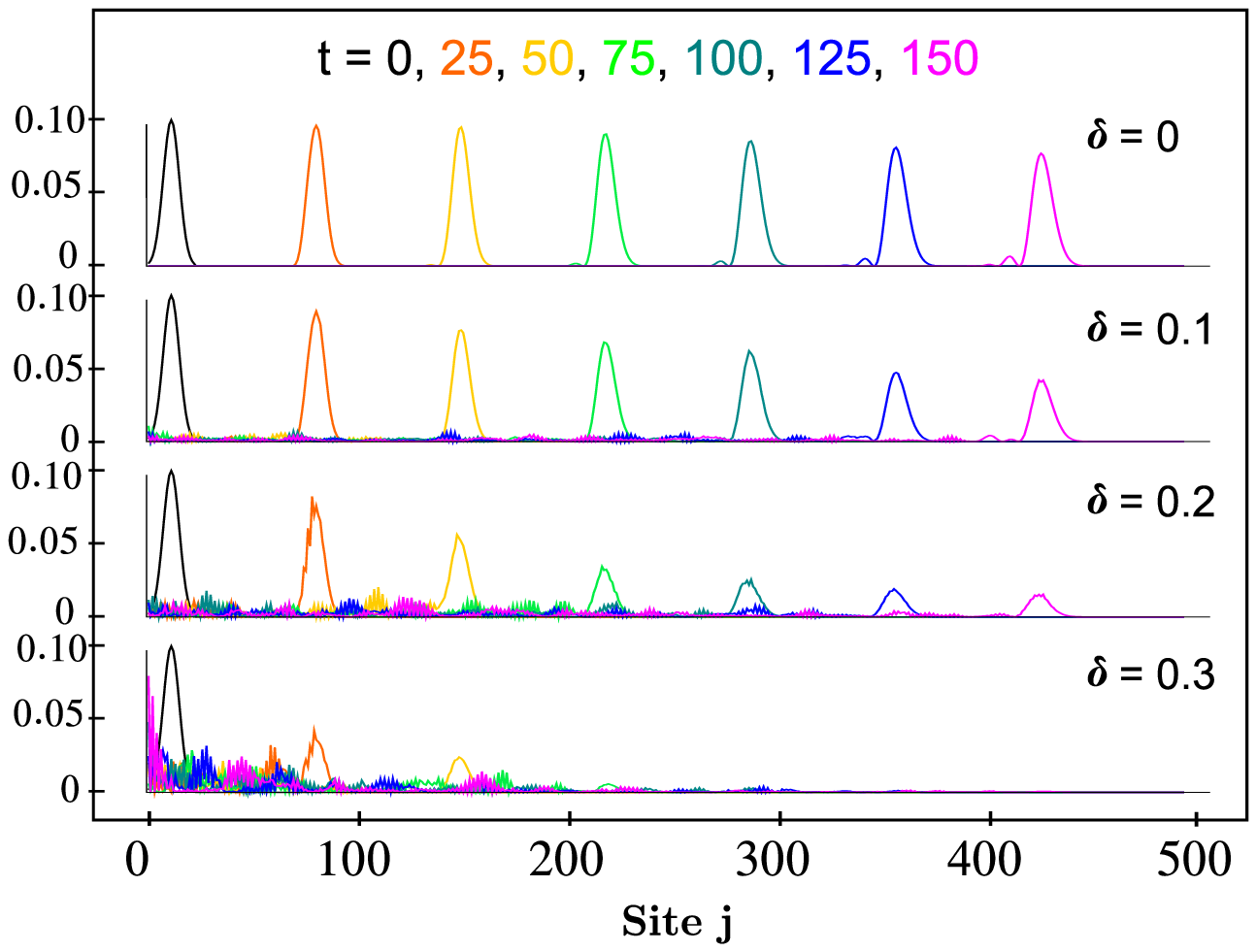}
        \put(-205,70){$\left|c_{j}\right|^2$}
    \caption{Effect of localization on state propagation for a disordered
    Heisenberg chain with $N = 501$, $J = 1/\sqrt{2}$. The wavepackets have been plotted at time increments $\delta t = 25$.
    Diagonal disorder was drawn from a normal distribution with mean zero
    and standard deviation $\delta$.  From top to bottom, $\delta = 0$, $0.1$, $0.2$ and $0.3$.}\label{fig:locprop}
\end{figure}
With a disordered chain, the channel parameter $
\mathcal{C}_{B,\delta}(t)$ depends on both the time and the
particular realization of the disorder (that is, the specific
values of the $\epsilon_{j}$ in \eqref{e:perturbed}). However,
for a given $\mathcal{P}_{\delta}$, the specific values of
$\epsilon_{j}$ can only be determined probabilistically. In other
words, $\gamma \equiv \gamma_{\delta}(t) = 1 -
\mathcal{C}_{B,\delta}(t)$ is a stochastic function of time,
parameterized by $\delta$.

We have investigated this claim in more depth by numerically
evaluating $\overline{\gamma_{\delta}}\left(t\right)$ for various
values of $t$ and $\delta$, with $\mathcal{P}_{\delta}$ chosen to
be the uniform distribution. Since the outcome is stochastic, we
average over many trials in order to build up a mean surface (see
Fig~\ref{fig:multitrials}). This was found to have the empirical form
\begin{equation}\label{e:emp}
    \overline{\gamma_{\delta}}\left(t\right) = 1 - e^{-\alpha t\left(\delta^{2} + \beta\delta\right)}
\end{equation}
for some constants $\alpha$ and $\beta$. The same empirical form \eqref{e:emp} was found to hold for
$\mathcal{P}_{\delta}$ chosen to be Cauchy and normal
distributions, although different values of $\alpha$ and $\beta$
were observed. Thus, if we know the value of $\delta$, we can
deduce how far Bob can be from Alice before, on average, the
fidelity drops below a certain threshold.  As expected, for fixed
$\delta$ the fidelity decays exponentially quickly as the state
propagates along the chain.

\begin{figure}
        \includegraphics[scale = 0.5]{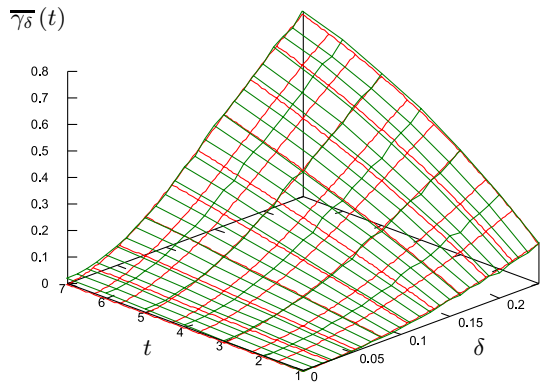}
        \put(-25,10){$\delta$}
        \put(-150,10){$t$}
        \put(-200,135){$\overline{\gamma_{\delta}}\left(t\right)$}
        \caption{Mean surface of $\gamma\left(t,\delta\right)$. Red: numerical data. Green: $1 - e^{-\alpha
        t\left(\delta^{2} + \beta\delta\right)}$, for $\alpha = 2.56$, $\beta = 0.029$.
        Diagonal disorder drawn from a uniform distribution in the
        range $\left[-\delta, \delta\right]$ .}\label{fig:multitrials}
\end{figure}

The identification of a disordered spin chain as an amplitude
damping channel (albeit one with a stochastic damping parameter)
immediately opens up the possibility of using quantum error
correction to improve the channel fidelity. Alice can encode her message state into a state of multiple qubits and
send each of these qubits down a separate, parallel, spin chain. 
Bob will then receive the multiple qubits and apply
standard error correction techniques.  Indeed, quantum codes
specifically tailored to treating amplitude damping errors are
already known \cite{LNCY97, FSW07}. However, if Alice and
Bob are separated by a distance larger than the localization
length, the amplitude damping will become too severe and this
error correction protocol will fail \footnote{In one dimension, for any
non-zero degree of disorder, the localization length is finite
with probability one.}.

However we have the option of performing error correction at regular intervals (shorter than the
localization length) along the chain, intervening before the
amplitude damping becomes too great. Thus as long as $\delta$ is not too large,  error correction can
be used to correct localization errors, and achieve high fidelity
state transfer over distances much larger than the localization
length. Furthermore, since the wavepackets propagate with a well-defined group velocity, the transmission takes a time linear in the 
distance over which we want to propagate.  

The key result of this Letter is that this scheme is scalable.
That is, the number of parallel chains needed to faithfully
communicate a qubit grows favourably with the distance we wish to
send it.   We show this by considering the following protocol. Alice encodes her initial qubit in the space of $n = 5^{k}$ qubits
according to the 5 qubit code \cite{BDSW96,LMPZ96} concatenated
$k$ times. Each qubit is then sent down a separate spin chain  -
modified via a channel twirling process \cite{HHH99} - encoded
over Alice's $N_A$ sites as a minimally-dispersive wavepacket in
accordance with \cite{OL04} and \cite{Has05}. (The twirling is applied
to ensure that the concatenation preserves the structure of the
channel.) At periodic intervals of distance $L$ we will error
correct, until the wavepackets reach Bob's end of the chain.
Bob then decodes each of the (distorted) wavepackets back down to the
space of single qubits, before finally decoding these $5^{k}$
qubits down to the space of a single qubit.

Provided that the number of parallel chains $n$ scales
polylogarithmically in the distance we wish to communicate over,
we show that it is possible to send a qubit an arbitrary distance,
with arbitrarily high fidelity using this protocol. More
specifically, let $L$ be a non-zero length chosen so that after
twirling, each length $L$ of the chain has the depolarizing parameter $p$. (For
the time being we will consider, for convenience, that each channel has
the same value of $p$; we will discuss at the end of the Letter how to relax this). Let $m$ be a positive integer, and let
$\epsilon$ be a positive constant in the interval
$\left(0,1\right]$. Then it is possible to send a qubit, a
distance $mL$,  with fidelity greater than $1-\epsilon/2$, using a
number of chains $n$ polylogarithmic in $m$:
\begin{equation}\label{e:polylogbound}
n=5^k > \left(\ln\frac{m}{\epsilon}\right)^3
\left(-\ln\frac{15p}{2}\right)^{-3}.
\end{equation}

To show this, we need the following facts regarding quantum channels:

\begin{lemma}\label{L:twirl} \textbf{Channel Twirling} \cite{HHH99}. Any channel $\Lambda$ can be turned into a depolarizing channel with the same entanglement fidelity F and average channel fidelity f as $\Lambda$.
\end{lemma}

A \textit{depolarizing channel with parameter p} is a quantum channel which, with probability $(1-p)$ leaves an input state $\rho$ untouched, and with probability $p$ replaces $\rho$ with the maximally mixed state. 
%%%%%%%%%%%%%%%%%%%%%%%%%%%%%%%%%%%%%%%%%%%%%%%%%%%%%%%%%%%%%%%%%%%%%%%%%%%%%%%%%%%%%%%%%%%%%%%%
\begin{lemma}\label{L:concatenate}
Consider the following procedure:
(a) Encoding a single qubit in the space of $5^{k}$ qubits according to the 5 qubit code concatenated $k$ times.
(b) Sending each qubit down an identical depolarizing channel with parameter $p$.
(c) Error correcting and decoding back to the space of a single qubit.
This procedure is equivalent to sending the qubit down a depolarizing channel with parameter $p_{k} < \frac{2}{15}\left(\frac{15}{2}p\right)^{2^{k}}$.
\end{lemma}
\begin{proof} Consider the simpler procedure:
(a) Encoding a single qubit in the space of 5 qubits according to the 5 qubit code.
(b) Sending each qubit down an identical depolarizing channel with parameter $p$.
(c) Error correcting and decoding back to the space of a single qubit.
Explicit calculation shows that this procedure is equivalent to sending the single qubit down a depolarizing channel with parameter $\frac{15}{2}p^{2} - \frac{25}{2}p^{3} + \frac{15}{2}p^{4} -\frac{3}{2}p^{5}$, which, in our range of interest $0 \leq p \leq 1$,  is monotonically non-decreasing in $p$ and (for $p$ not zero) strictly less than $\frac{15}{2}p^{2}$.  The effect of concatenating our quantum code then follows by induction. 
\end{proof}
%%%%%%%%%%%%%%%%%%%%%%%%%%%%%%%%%%%%%%%%%%%%%%%%%%%%%%%%%%%%%%%%%%%%%%%%%%%%%%%%%%%%%%%%%%%%%%%%
\begin{lemma}\label{L:compose}
Composing $m$ identical depolarizing channels, each with channel
parameter $p$, gives a depolarizing channel with parameter
$1-(1-p)^{m}$.
\end{lemma}

%%%%%%%%%%%%%%%%%%%%%%%%%%%%%%%%%%%%%%%%%%%%%%%%%%%%%%%%%%%%%%%%%%%%%%%%%%%%%%%%%%%%%%%%%%%%%%%%
Suppose we have at our disposal many quantum channels, each with
physical length $L$.  Let us call these \textit{basic channels.}
By lemma \ref{L:twirl} we can assume, without loss of generality,
that these are depolarizing channels with parameter $p$. Now
consider combining $5^{k}$ of these channels in parallel to form a
block of length $L$ and \textit{depth} (i.e. number of channels in
parallel) $5^{k}$. By lemma \ref{L:concatenate} we can use error
correction to send a qubit down this block, and the net effect is
equivalent to sending the qubit down a depolarizing channel with
parameter $p_{k}$.  Let us call this resulting channel a
\textit{k-block}, in view of the fact that $k$ levels of
concatenation are involved. If we now compose $m$ of these blocks
together to form a channel of depth $5^{k}$ and total length $mL$,
lemma \ref{L:compose} tells us that this is equivalent to a
depolarizing channel with parameter
\begin{equation}
    p_{\text{total}} < 1 - \left[1-\frac{2}{15}\left(\frac{15}{2}p\right)^{2^{k}}\right]^{m}. \label{E:ptotal}
\end{equation}

Now, (\ref{e:polylogbound}) means that
\begin{eqnarray}
5^k &>& \left(\frac {\ln(\frac{2}{15}(1+\frac{m}{\epsilon}))}
{-\ln{\frac{15p}{2}}}\right)^{\log_2 5}\nonumber\\
\Rightarrow 2^k  &>&
\left(
\frac
{\ln(\frac{15}{2}(1-e^{-\epsilon/m})}
{\ln \frac{15p}{2}}
\right).\label{e:first-ineq}
\end{eqnarray}
\begin{eqnarray}
\Rightarrow \frac{2}{15}\left(\frac{15}{2}p\right)^{2^{k}} &<& 1-e^{-\epsilon/m}
\nonumber\\
\Rightarrow p_{\text{total}} &<& 1-e^{-\epsilon}\nonumber\\
&<& \epsilon. \label{e:second-ineq}
\end{eqnarray}
We have used the  inequalities  $1-e^{-x}> \frac{x}{1+x}$ for $x>0$ and $1-e^{-x}< x$ for $x>0$ in (\ref{e:first-ineq}) and (\ref{e:second-ineq}) respectively.

So provided that the parameter $p$ of the fundamental depolarizing
channels is less than $2/15$, we are free to choose any $m$,
arbitrarily large, and any $\epsilon$, however close to zero, and
the overall channel will have parameter
$p_{\text{total}}<\epsilon$ as long as  the channel depth
satisfies (\ref{e:polylogbound}). The average fidelity of sending a qubit down this channel is then
$f > 1 - \epsilon/2$.

We can now directly apply this result to the case where our basic channels are the sections of the disordered spin chains. As mentioned before, each chain acts like an amplitude damping channel with identical parameter $\gamma$. If we now apply channel twirling to our chains they become depolarizing channels with parameter $p = \frac{1}{3}\left[\gamma + 2\left(1-\sqrt{1-\gamma}\right)\right]$. We require that $p < 2/15$ so, given a degree of disorder $\delta$, we choose a value of $L$ such that this is true, at least with high probability.  In other words, the degree of disorder fixes the maximum possible length of our basic channels, or equivalently, the length before which we must error correct.  We can now form $k$-blocks from our twirled chains, and compose $m$ of these blocks together to form a channel of total length $mL$.  It then follows that it is possible, using a number of disordered chains polylogarithmic in $m/\epsilon$, to send a qubit a distance of $mL$ with overall fidelity greater than $1-\epsilon/2$.

%%%%%%%%%%%%%%%%%%%%%%%%%%%%%%%%%%%%%%%%%%%%%%%%%%%%%%%%%%%%%%%%%%%%%%%%%%%%%%%%%%%%%%%%%%%%%%%%

We have shown that for a class of spin chain protocols, namely
those that satisfy conditions $(\text{C1})$-$(\text{C5})$, that
localization can be viewed as a source of amplitude damping
errors. Although we have focused here on the case of diagonal
disorder, the same holds true for more general models of disorder,
provided that the disorder does not couple together subspaces of
different numbers of excitations. Indeed it will be realized
that our protocol can be used to deal with rather more general
errors arising in quantum communication in spin chains and not
just those arising from localization.

In the text we have considered that each channel has an identical
error parameter. This is, of course, an unreasonable assumption
since the errors are stochastic.  However it is not too difficult
to see that the protocol still leads to polylogarithmic scaling
if, for example, the channels are guaranteed to have error parameter
$p$ below some given threshold. It would be an interesting question for the future to analyse how well the protocol succeeds if one only knows, for example, the distribution of possible channel parameters..   

It is worth making explicit the nature of the interventions needed
along the chain in order to perform our protocol.  The syndrome
measurements may be done coherently, rather than by performing von
Neumann measurements.  In other words, the individual chains can have
(passive) unitary interactions between them that perform the error
correction.  This requires clean qubits -- sources of low entropy -- to
be coupled into the chains at regular intervals.

Finally we observe that our results might be interesting in the
context of solid state physics.  We have shown that parallel
disordered one-dimensional spin chains can support high fidelity
\lq\lq conduction\rq\rq\ of quantum information over arbitrary
distances; the number of required chains scaling only
polylogarithmically with the distance.  This is perhaps at odds
with the intuition one might have from the well-known fact that in
one dimension, disorder inevitably prevents propagation.  Of
course there is no true contradiction here. For example, our system is not
strictly one dimensional; the error correction leads to subtle
coupling of the one-dimensional chains. However we believe that
our techniques from quantum information may offer new insights
into localization in solid state systems. 

We would like to thank Toby Cubitt and Andreas Winter for many
helpful discussions.  J.A. also gratefully acknowledges the
support of the Dorothy Hodgkin Foundation.

\end{document}